\documentclass[runningheads,a4paper]{llncs}

\usepackage{amssymb,amsmath,tikz}
\usepackage{graphicx}

\usepackage{url}
\usepackage{enumerate}
\newcommand{\keywords}[1]{\par\addvspace\baselineskip
\noindent\keywordname\enspace\ignorespaces#1}

\makeatletter
\newcommand{\text@hyphens}{\mathcode`\-=`\-\relax}
\newcommand{\id}[1]{\ensuremath{\mathit{\text@hyphens#1}}}
\makeatother

\begin{document}

\mainmatter  

\title{Individual and Group Stability in Neutral Restrictions of Hedonic Games}

\titlerunning{Individual and Group Stability in Neutral Restrictions of Hedonic Games}

%
%
\author{
Warut Suksompong%
\thanks{Supported by a Stanford Graduate Fellowship.}%
}
\authorrunning{W. Suksompong}

\institute{Department of Computer Science, Stanford University\\
353 Serra Mall, Stanford, CA 94305, USA\\
\email{warut@cs.stanford.edu}\\
}

%
%

\maketitle

\begin{abstract}
We consider a class of coalition formation games called \textit{hedonic games}, i.e., games in which the utility of a player is completely determined by the coalition that the player belongs to. We first define the class of \textit{subset-additive hedonic games} and show that they have the same representation power as the class of hedonic games. We then define a restriction of subset-additive hedonic games that we call \textit{subset-neutral hedonic games} and generalize a result by Bogomolnaia and Jackson (2002) by showing the existence of a Nash stable partition and an individually stable partition in such games. We also consider \textit{neutrally anonymous hedonic games} and show that they form a subclass of the subset-additive hedonic games. Finally, we show the existence of a core stable partition that is also individually stable in neutrally anonymous hedonic games by exhibiting an algorithm to compute such a partition. 
\end{abstract}

\section{Introduction}

Coalition formation plays a major role in a broad range of settings. Whether politicians forming political parties to participate in an election, countries joining alliances to increase their negotiation power, or students getting together in groups for a classroom project, one can observe coalition formation at work. It is therefore important to understand the process of how coalitions form: when can we expect coalitions to form, and when are the coalition members ``happy'' with their coalitions? Within game theory, the importance of coalitions became clear since von Neumann and Morgenstern's seminal work \textit{Theory of Games and Economic Behavior} was published in 1944.

The hedonic viewpoint of coalition formation was introduced by Dr\`{e}ze and Greenberg (1980), who first called attention to the ``hedonic aspect'' of the game, i.e., the determination of a player's utility by the coalition that the player belongs to. Banerjee et al. (2001) and Bogomolnaia and Jackson (2002) introduced and analyzed stability concepts in hedonic coalition formation games, including \textit{Nash stability}, \textit{individual stability}, and \textit{core stability}. Among other things, Bogomolnaia and Jackson defined restricted classes of hedonic games called \textit{additively separable hedonic games} and \textit{symmetric additively separable hedonic games}, and showed that a Nash stable partition always exists in the latter class of games but not necessarily in the former. Since then, several restrictions on hedonic games have been proposed, including hedonic games based on best players and hedonic games based on worst players by Cechl\'{a}rov\'{a} and Romero-Medina (2001) and fractional hedonic games by Aziz et al. (2014). Other authors who have proposed and analyzed restrictions on hedonic games include Alcalde and Revilla (2004), Alcalde and Romero-Medina (2006), Burani and Zwicker (2003), and Dimitrov et al. (2006). For an excellent survey on hedonic games, we refer to Hajdukov\'{a} (2006).

In this paper, we define the class of \textit{subset-additive hedonic games}, which generalizes the class of additively separable hedonic games. We show that our class does not provide any restriction on the game -- any hedonic game is also a subset-additive hedonic game. We then define the class of \textit{subset-neutral hedonic games}, which is a restriction of subset-additive hedonic games and a generalization of symmetric additively separable hedonic games. Even though subset-neutral hedonic games enjoy significantly more representation power than symmetric additively separable hedonic games, they also provide a guarantee of the existence of a Nash stable partition and an individually stable partition. We also consider \textit{neutrally anonymous hedonic games}, which is a somewhat restricted class of games but can still model many interesting situations. We show that they form a subclass of the subset-neutral hedonic games, hence inheriting the guarantee of the existence of a Nash stable partition and an individually stable partition. Finally, we show that a core stable partition that is also individually stable is guaranteed to exist in neutrally anonymous hedonic games by exhibiting an algorithm that computes such a partition.

\section{Definitions and notation}

In this section, we introduce the setting and give definitions and notation that we will use throughout this paper.

Let $N=\{1,2,\ldots,n\}$ be a finite set of $n$ players in the game. Denote by $N_i$ the set of all subsets of $N$ that include $i$. A \textit{coalition} is a nonempty subset of $N$. A \textit{coalition partition} $\pi$ is a partition of the set $N$ into disjoint coalitions. Denote by $C_\pi(i)$ the coalition in $\pi$ that $i$ belongs to. For any set $S$, denote by $2^S$ the set of all subsets of $S$ and $|S|$ the size of $S$.

We assume throughout the paper that the players' preferences are \textit{hedonic}, i.e., they are completely determined by the coalition that the player belongs to. Each player $i$ is endowed with a preference relation $\succeq_i$, a reflexive, complete, and transitive ordering over $N_i$. Let $\succ_i$ denote the strict part and $\sim_i$ the indifference part of the relation $\succeq_i$. A \textit{hedonic coalition formation game} is represented by a pair $(N, \{\succeq_i\}_{i=1}^n)$.

We now define stability notions that we will consider in the paper. The following three stability notions, which consider deviations by a single player, were introduced by Bogomolnaia and Jackson (2002). 

\begin{definition}
A coalition partition $\pi$ is \textit{Nash stable} if $C_\pi(i)\succeq_i C\cup\{i\}$ for all $i\in N$ and $C\in\pi\cup\{\emptyset\}$.
\end{definition}

\begin{definition}
A coalition partition $\pi$ is \textit{individually stable} if for any $i\in N$ and $C\in\pi\cup\{\emptyset\}$ such that $C\cup\{i\}\succ_i C_\pi(i)$, there exists a player $j\in C$ such that $C\succ_j C_k\cup\{i\}$.
\end{definition}

\begin{definition}
A coalition partition $\pi$ is \textit{contractually individually stable} if for any $i\in N$ and $C\in\pi\cup\{\emptyset\}$ such that $C\cup\{i\}\succ_i C_\pi(i)$, there exists either a player $j\in C$ such that $C\succ_j C_k\cup\{i\}$, or a player $k\in C_\pi(i)$ such that $k\neq i$ and $C_\pi(i)\succ_{k} C_\pi(i)\backslash\{i\}$.
\end{definition}

Any Nash stable partition is also individually stable and any individually stable partition is also contractually individually stable. Ballester (2004) showed that every hedonic game contains at least one contractually individually stable partition.

Next, we define two stability notions that consider deviations by a coalition. 

\begin{definition}
A coalition $C$ \textit{blocks} a coalition partition $\pi$ if $C\succ_i C_\pi(i)$ for all $i\in C$. A coalition partition $\pi$ is \textit{core stable} if any coalition $C\subseteq N$ does not block $\pi$.
\end{definition}

\begin{definition}
A coalition $C$ \textit{weakly blocks} a coalition partition $\pi$ if $C\succeq_i C_\pi(i)$ for all $i\in C$ and $C\succ_j C_\pi(j)$ for some $j\in C$. A coalition partition $\pi$ is \textit{strong core stable} if any coalition $C\subseteq N$ does not weakly block $\pi$.
\end{definition}

Any strong core stable partition is also core stable as well as individually stable, and neither Nash stability nor core stability implies the other.

We now define properties on preference profiles. 

\begin{definition}
A preference profile $\{\succeq_i\}_{i=1}^n$ is \textit{separable} if
\begin{itemize}
\item $C\cup\{j\}\succeq_i C$ if and only if $\{i,j\}\succeq_i \{i\}$, and
\item $C\cup\{j\}\succ_i C$ if and only if $\{i,j\}\succ_i \{i\}$
\end{itemize}
for all $i,j\in N$ such that $i\in C$ and $j\not\in C$.
\end{definition}

\begin{definition}
A preference profile $\{\succeq_i\}_{i=1}^n$ is \textit{additively separable} if for all $i\in N$, there exists a function $v_i:N\rightarrow\mathbb{R}$ such that \[C\succeq_i D \text{ if and only if } \sum_{j\in C}v_i(j)\geq\sum_{j\in D}v_i(j)\]
for all $C,D\in N_i$.

Furthermore, an additively separable preference profile is called \textit{symmetric} if $v_i(j)=v_j(i)$ for all $i,j\in N$, and is called \textit{mutual} if for all $i,j\in N$, $v_i(j)\geq 0$ whenever $v_j(i)\geq 0$.
\end{definition}



\begin{definition}
A preference profile $\{\succeq_i\}_{i=1}^n$ is \textit{anonymous} if $C\sim_i D$ for any $i\in N$ and any two coalitions $C,D$ such that $i\in C,D$ and $|C|=|D|$.
\end{definition}

\begin{definition}
A preference profile $\{\succeq_i\}_{i=1}^n$ satisfies the \textit{common ranking property} if there exists a function $w:2^N\backslash\{\emptyset\}\rightarrow\mathbb{R}$ such that for all $i\in N$,
\[C\succeq_i D \text{ if and only if } w(C)\geq w(D)\]
for all $C,D\in N_i$.
\end{definition}

\begin{definition}
Given a nonempty set $V\subseteq N$, a nonempty subset $S\subseteq V$ is a \textit{top coalition} of $V$ if for any $i\in S$ and any $T\subseteq V$ with $i\in T$, we have $S\succeq_i T$. A preference profile $\{\succeq_i\}_{i=1}^n$ satisfies the \textit{top-coalition property} if for any nonempty set $V\subseteq N$, there exists a top coalition of $V$.
\end{definition}

\section{Subset-additive and subset-neutral hedonic games}
\label{sec:subaddsubneu}

In this section, we define a generalization of additively separable games that we call \textit{subset-additive hedonic games}, and show that the class of subset-additive hedonic games in fact coincides with the class of hedonic games. We then define a generalization of symmetric additively separable games that we call \textit{subset-neutral hedonic games}. We show that subset-neutral hedonic games have more representation power than symmetric additively separable games, and we generalize a result by Bogomolnaia and Jackson (2002) by proving that the existence of a Nash stable partition and an individually stable partition is guaranteed in subset-neutral hedonic games.

We start with the definition of subset-additive preference profiles.

\begin{definition}
A preference profile $\{\succeq_i\}_{i=1}^n$ is \textit{subset-additive} if for all $i\in N$, there exists a function $v_i:N_i\rightarrow\mathbb{R}$ such that 
\[C\succeq_i D \text{ if and only if } \sum_{i\in C'\subseteq C}v_i(C')\geq\sum_{i\in D'\subseteq D}v_i(D').\]
\end{definition}

It turns out that subset-additivity does not provide a restriction on the preference profile, as the following proposition shows.

\begin{proposition}
\label{prop:subsetadd}
Any hedonic game is also a subset-additive hedonic game.
\end{proposition}

\begin{proof}
Consider a preference profile $\{\succeq_i\}_{i=1}^n$ of a hedonic game. For all $i\in N$, let $u_i:N_i\rightarrow\mathbb{R}$ be a utility function consistent with player $i$'s preference profile. We define the function $v_i$ recursively from smaller to larger sets in $N_i$. If $v_i(C')$ has been defined for all $\{i\}\in C'\subsetneq C$, we define 
\[v_i(C)=u_i(C)-\sum_{i\in C'\subsetneq C}v_i(C'). \]
It follows that 
\[u_i(C)\geq u_i(D) \text{ if and only if } \sum_{i\in C'\subseteq C}v_i(C')\geq\sum_{i\in D'\subseteq D}v_i(D'),\]
which implies that the game is subset-additive, as desired.
\end{proof}

Although subset-additivity provides no restriction on the preference profile, if we impose a neutrality condition on the utility function, we obtain a smaller class of preference profiles.

\begin{definition}
\label{def:subsetneu}
A preference profile $\{\succeq_i\}_{i=1}^n$ is \textit{subset-neutral} if there exists a function $w:2^N\backslash\{\emptyset\}\rightarrow\mathbb{R}$ such that  
\[C\succeq_i D \text{ if and only if } \sum_{i\in C'\subseteq C}w(C')\geq\sum_{i\in D'\subseteq D}w(D')\]
for all $i\in N$.
\end{definition}

Subset-neutral preference profiles are useful for modeling situations in which different teams can form within a coalition. For example, suppose that a police department is divided into different subdivisions, which correspond to our coalitions. Certain teams of police officers will be assigned by the chief to tackle a criminal case based on their combined specialty if they belong to the same subdivision. The value of a subdivision to a police officer is therefore the sum of the values of the different teams to which he will be assigned to work on cases. Since team chemistry varies according to the composition of the team, one can imagine that the value is different for different teams and cannot be broken down into values between pairs as in additively separable preference profiles. Subset-neutral preference profiles also allow for situations in which the chief assigns as many or as few teams as he likes, since teams that are not assigned simply correspond to a value of 0.

If we set $w(C)=0$ for all $|C|>2$ in Definition \ref{def:subsetneu}, we obtain the class of symmetric additively separable preference profiles. Hence any symmetric additively separable preference profile is also subset-neutral. On the other hand, not all subset-neutral profiles are additively separable (or even separable), as the following example shows.

\begin{example}
Consider the game with $N=\{1,2,3\}$ and the function $w$ given by 
\begin{itemize}
\item $w(\{i\})=0$ for all $i\in N$;
\item $w(\{1,2\})=w(\{1,3\})=w(\{2,3\})=1$;
\item $w(\{1,2,3\})=-10$. 
\end{itemize}
We have $\{1,3\}\succ_1\{1\}$ and $\{1,2,3\}\prec_1\{1,2\}$, which violates separability.
\end{example}

A preference profile satisfying subset neutrality cannot have cycles on coalitions of size 2 between different players. For instance, it cannot be the case that $\{1,2\}\succ_1\{1,3\}, \{1,3\}\sim_3\{2,3\}$, and $\{2,3\}\succ_2\{1,2\}$ hold simultaneously. On the other hand, cycles on coalitions of size greater than 2 between different players can occur, as the following example shows. 

\begin{example}
Consider the game with $N=\{1,2,3,4\}$ and the function $w$ given by 
\begin{itemize}
\item $w(\{i\})=0$ for all $i\in N$;
\item $w(\{1,2\})=w(\{1,2,3\})=w(\{1,2,4\})=w(\{1,4\})=w(\{2,3\})=0$; 
\item $w(\{1,3\})=w(\{2,4\})=1$. 
\item The value of $w$ on other subsets can be defined arbitrarily. 
\end{itemize}
We have $\{1,2,3\}\succ_1\{1,2,4\}$ and $\{1,2,4\}\succ_2\{1,2,3\}$.
\end{example}

Bogomolnaia and Jackson (2002) showed that for symmetric additively separable hedonic games, a Nash stable partition and an individually stable partition exist. The next theorem generalizes that result.

\begin{theorem}
\label{thm:neunash}
For subset-neutral hedonic games, a Nash stable partition and an individually stable partition exist.
\end{theorem}

\begin{proof}
Since a Nash stable partition is also individually stable, it suffices to show that a Nash stable partition exists.

Consider a partition $\pi$ that maximizes the potential function \[\Phi(\pi)=\sum_{C\in\pi}\sum_{\emptyset\neq C'\subseteq C}w(C');\] such a partition exists since the total number of partitions is finite. We claim that $\pi$ is a Nash stable partition. Suppose for contradiction that player $i$ has an incentive to move from coalition $C_\pi(i)$ to a different coalition $C_j$. This means that the utility that $i$ gains from the subsets in $C_j$ is greater than the utility that $i$ gains from the subsets in $C_\pi(i)$, i.e., \[\sum_{i\in D\subseteq C_j\cup\{i\}}w(D)>\sum_{i\in D\subseteq C_\pi(i)}w(D).\]

Let $\pi'$ be the partition that is obtained from $\pi$ if $i$ moves from coalition $C_\pi(i)$ to $C_j$. The potential function of the partition $\pi'$ is
\begin{align}
\Phi(\pi')  &= \Phi(\pi)+\sum_{i\in D\subseteq C_j\cup\{i\}}w(D)-\sum_{i\in D\subseteq C_\pi(i)}w(D) \notag \\
            &> \Phi(\pi), \notag 
\end{align}
contradicting the assumption that $\pi$ is a partition that maximizes the potential function.
\end{proof}

The proof of Theorem \ref{thm:neunash} relies crucially on neutrality. Bogomolnaia and Jackson (2002) showed that an individual stable partition (and hence a Nash stable partition) may not exist even if preference profiles are additively separable, mutual, and single peaked on a tree. (For the definition of single-peakedness on a tree, we refer to their paper.) On the other hand, symmetry in additively separable hedonic games is not enough to guarantee the existence of a core stable partition. Indeed, Banerjee et al. (2001) showed that a core stable partition (and hence a strong core stable partition) may not exist even if preference profiles are additively separable and symmetric.

\section{Neutral anonymity}

In this section, we define a restriction of anonymous hedonic games that we call \textit{neutrally anonymous hedonic games}. We show that such games are subset-neutral, and hence existence of a Nash stable partition and an individually stable partition is guaranteed by results in Section \ref{sec:subaddsubneu}. We then show that a partition that is both core stable and individually stable must exist in neutrally anonymous hedonic games, and in fact, in the more general class of games whose preference profiles satisfy the common ranking property. We exhibit an algorithm to find such a partition.

We start with the definition of neutrally anonymous preference profiles.

\begin{definition}
\label{def:neuanon}
A preference profile $\{\succeq_i\}_{i=1}^n$ is \textit{neutrally anonymous} if there exists a function $f:\{1,2,\ldots,n\}\rightarrow\mathbb{R}$ such that for all $i\in \{1,2,\ldots,n\}$,
\[C\succeq_i D \text{ if and only if } f(|C|)\geq f(|D|).\]
\end{definition}

Note that the function $f$ in the definition takes on the cardinality of the set of players rather than the set of players itself.

Any neutrally anonymous preference profile is also anonymous and satisfies the common ranking property. On the other hand, a neutrally anonymous preference profile need not be separable or single-peaked on the size of the coalition to which the player belongs. Indeed, consider the neutrally anonymous hedonic game with $N=\{1,2,3\}$ and the function $f$ given by $f(1)=0$, $f(2)=1$ and $f(3)=-1$. This preference profile is neither separable nor single-peaked on the size of the coalition to which the player belongs.

Even though neutral anonymity is somewhat restrictive, there are interesting situations that can be modeled using preference profiles satisfying this property. For instance, a teacher may assign a different amount of work in a classroom project to groups of students of different sizes. If the students are only concerned about the amount of work that they have to do, their utility will only depend on the size of their group. As another example, a restaurant may provide a different amount of food or issue a different price depending on the size of the party. If the utility of each customer is determined by the food-to-price ratio, the situation can again be modeled with a neutrally anonymous hedonic game. One can imagine that in such situations, the preference profile need not be monotonic or even single-peaked in the size of the coalition to which the player belongs.

It turns out that any neutrally anonymous hedonic game is also a subset-neutral hedonic game, as the following proposition shows.

\begin{proposition}
\label{prop:neuanon}
Any neutrally anonymous preference profile is also subset-neutral.
\end{proposition}

\begin{proof}
Consider an arbitrary neutrally anonymous game, and let $f$ be the function associated to it as in Definition \ref{def:neuanon}. We will exhibit a function $w$ associated to it as in Definition \ref{def:subsetneu} such that 
\[f(|C|) = \sum_{i\in C'\subseteq C}w(C')\]
for all $i\in C\subseteq N$.

We define the function $w$ recursively from smaller to larger subsets of $N$. Because of neutral anonymity, sets of the same size have the same value of the function $w$. Hence it suffices to define the function $w$ for the sets $\{1\},\{1,2\},\ldots,\{1,2,\ldots,n\}$. If $w(C)$ has been defined for all $C\subseteq N$ such that $|C|<k$, we define 
\[w(\{1,2,\ldots,k\})=f(k)-\sum_{1\in C\subsetneq\{1,2,\ldots,k\}}w(C).\]
It follows that
\[f(|C|) = \sum_{i\in C'\subseteq C}w(C')\]
for all $i\in C\subseteq N$, as desired.
\end{proof}

We obtain the following theorem from Theorem \ref{thm:neunash} and Proposition \ref{prop:neuanon}.

\begin{theorem}
\label{thm:neuanonexistence}
For neutrally anonymous hedonic games, a Nash stable partition and an individually stable partition exist.
\end{theorem}

In contrast to Theorem \ref{thm:neuanonexistence}, Bogomolnaia and Jackson (2002) showed that a Nash stable partition may not exist even if the preference profile is anonymous and single-peaked. The next example shows that a Nash stable partition may not exist even if the preference profile satisfies the common ranking property.

\begin{example}
Consider the game with $N=\{1,2,3\}$ and the common ranking given by \[\{1,2\}\succ\{1\}\succ\{2\}\succ\{1,2,3\}\succ\{1,3\}\succ\{2,3\}\succ\{3\}.\] 

Let $\pi$ be any partition of $N$. If player 3 is left alone in $\pi$, she will have an incentive to join one of the existing coalitions. Otherwise, player 3 is in the same coalition in $\pi$ as at least one of player 1 and player 2, and that player has an incentive to form a coalition by herself.
\end{example}

The next theorem shows that a core stable partition that is also individually stable and contractually individually stable exists in games whose preference profiles satisfy the common ranking property. On the other hand, a strict core stable partition may not exist even in neutrally anonymous hedonic games, as shown by the game with $N=\{1,2,3\}$ and the function $f$ given by $f(1)=0$, $f(2)=1$ and $f(3)=0$.

\begin{theorem}
\label{thm:neuanoncoreexistence}
For games whose preference profiles satisfy the common ranking property, a core stable partition that is also individually stable and contractually individually stable exists. Moreover, we exhibit an algorithm to compute such a partition.
\end{theorem}

\begin{proof}
Since individual stability implies contractually individual stability, it suffices to give an algorithm to compute a core stable partition that is also individually stable. The algorithm operates as follows:
\begin{enumerate}
\item Initially, let $S$ be the set of all players.
\item Repeat the following until $S=\emptyset$:
\begin{itemize}
\item Choose a coalition $T\subseteq S$ that ranks highest in the common ranking of all the nonempty subsets of $S$. If there are several such coalitions, choose one with the largest size.
\item Remove the players in $T$ from $S$.
\end{itemize}
\end{enumerate}

The algorithm clearly terminates. We first show that the resulting partition is core stable. We prove by induction that when a coalition is formed by the algorithm, every player in the coalition is unwilling to participate in a blocking coalition. This is true for the first coalition formed, since the players in the coalition rank the coalition highest in their ranking. For any subsequent coalition formed, a player in the coalition rank the coalition highest among the coalitions that she can form using her coalition and the remaining players. By the induction hypothesis, no player in previous coalitions is willing to participate in a blocking coalition. Hence the players in the coalition formed also do not want to participate in a blocking coalition, completing the induction.

We now show that the resulting partition is individually stable. Consider a deviation by player $i$. Again, she is already in a coalition that ranks highest among the coalitions that she can form using her coalition and the remaining players. Hence she has no incentive to switch to a coalition formed after her. If she switches to a coalition formed before her, then since that coalition is a largest coalition that ranks highest when it is formed, the inclusion of player $i$ necessarily leaves the members of that coalition worse off. Finally, player $i$ has no incentive to form a coalition on her own if she is not already in a coalition by herself. Hence the partition is individually stable, as desired.
\end{proof}

The algorithm in Theorem \ref{thm:neuanoncoreexistence} is a specific version of the algorithm proposed by Banerjee et al. (2001) for finding a core stable partition when the preference satisfies the top-coalition property. While our algorithm requires a stronger condition, it produces a partition that is both core stable and individually stable, instead of only core stable. The crucial difference between the two algorithms is that our algorithm requires in Step 2 that if there are several coalitions that rank highest among the remaining coalitions, our algorithm chooses a coalition with the largest size. If we eliminate this requirement, the resulting partition is no longer guaranteed to be individually stable, as the following example shows.

\begin{example}
Consider the game with $N=\{1,2,3\}$ and the common ranking given by \[\{1,2\}\sim\{1,2,3\}\succ\{2\}\succ\{1\}\succ\{1,3\}\succ\{2,3\}\succ\{3\}.\] 

Initially, $\{1,2\}$ and $\{1,2,3\}$ are the coalitions that rank highest in the common ranking. If the algorithm chooses the coalition $\{1,2\}$ as the first coalition, the resulting partition will be $\{1,2\},\{3\}$. This partition is not individually stable, however, since player 3 can benefit by joining the other two players while not leaving them worse off.
\end{example}

In addition, if we assume only the top-coalition property, then Banerjee et al.'s algorithm (or any specific version of it) is not guaranteed to return an individually stable partition, as the following example shows.

\begin{example}
Consider the game with $N=\{1,2,3\}$ and the preference profile $\succeq_1,\succeq_2,\succeq_3$ given by 
\begin{itemize}
\item $\{1,2\}\sim_1\{1,2,3\}\succ_1\{1\}\succ_1\{1,3\}$;
\item $\{1,2\}\sim_2\{1,2,3\}\succ_2\{2\}\succ_2\{2,3\}$;
\item $\{1,3\}\sim_3\{2,3\}\succ_3\{1,2,3\}\succ_3\{3\}$.
\end{itemize}

The preference profile $\succeq$ satisfies the top-coalition property. Indeed, $\{i\}$ is a top coalition of $\{i\}$ for all $i\in N$, $\{1,2\}$ is a top coalition of $\{1,2\}$ and $\{1,2,3\}$, $\{1\}$ is a top coalition of $\{1,3\}$, and $\{2\}$ is a top coalition of $\{2,3\}$.

Since $\{1,2\}$ is the unique top coalition in $\{1,2,3\}$, Banerjee et al.'s algorithm (or any specific version of it) will choose $\{1,2\}$ in the first step, yielding the partition $\{1,2\},\{3\}$. This partition is not individually stable, however, since player 3 can benefit by joining the other two players while not leaving them worse off.
\end{example}

We obtain the following as a corollary of Theorem \ref{thm:neuanoncoreexistence}.

\begin{corollary}
For neutrally anonymous hedonic games, a core stable partition that is also individually stable and contractually individually stable exists.
\end{corollary}

In contrast to Theorem \ref{thm:neuanoncoreexistence}, Banerjee et al. (2001) showed that a core stable partition may not exist even if the preference profile is anonymous, singled-peaked on population, and satisfies the population's intermediate preference property. (For definitions, we refer to their paper.)

Although a core stable partition produced by the algorithm in Theorem \ref{thm:neuanoncoreexistence} is guaranteed to be individually stable, this is not necessarily the case for any core stable partition in neutrally anonymous hedonic games. In fact, a core stable partition need not even be contractually individually stable in such games, as the following example shows.

\begin{example}
Consider the game with $N=\{1,2,3\}$ and the function $f$ given by $f(1)=0$ and $f(2)=f(3)=1$. 

The partition $\{1,2\},\{3\}$ is core stable but not contractually individually stable, since player 3 can join the coalition $\{1,2\}$ without leaving any member of her old or new coalition worse off. In this example, the grand coalition $\{1,2,3\}$ is both core stable and individually stable.
\end{example}

Theorems \ref{thm:neuanonexistence} and \ref{thm:neuanoncoreexistence} guarantee the existence of both a Nash stable partition and a core stable partition in neutrally anonymous hedonic games. Nevertheless, there may not exist a partition that is both Nash stable and core stable, as the following example shows.

\begin{example} 
\label{ex:nonashcore}
Consider the game with $N=\{1,2,3,4,5\}$ and the function $f$ given by $f(1)=0$, $f(2)=2$, $f(3)=1$, and $f(4)=f(5)=0$. 

Any core stable partition cannot contain a coalition of size 3, 4, or 5, since two players from such a coalition form a blocking coalition. A partition that consists only of coalitions of size 1 and 2, however, necessarily contains a player in a coalition by herself. Such a player has an incentive to switch to any existing coalition, which implies that such a partition cannot be Nash stable. 
\end{example}

Example \ref{ex:nonashcore} also shows that neither Nash stability nor core stability implies the other even in neutrally anonymous hedonic games.




\end{document}